\documentclass[aps,prl,letterpaper,superscriptaddress,twocolumn]{revtex4}

\usepackage{amstext,amsmath,amssymb,ulem,amsthm}
\usepackage{color}

\usepackage{times}
\usepackage{algorithm}
\floatname{algorithm}{Protocol}
\usepackage{ftnxtra}

\usepackage[T1]{fontenc}
\usepackage[utf8]{inputenc}

\usepackage{graphicx}
\usepackage[colorlinks]{hyperref}

\theoremstyle{plain}

\theoremstyle{definition}
\newtheorem{definition}{Definition}

\newcommand{\beq}{\begin{equation}}
\newcommand{\eeq}{\end{equation}}

\newcommand{\ket} [1] {\vert #1 \rangle}
\newcommand{\bra} [1] {\langle #1 \vert}
\newcommand{\braket}[2]{\langle #1 | #2 \rangle}

\newcommand{\ba}{\begin{align}}
\newcommand{\ea}{\end{align}}
\newcommand{\bea}{\begin{eqnarray}}
\newcommand{\eea}{\end{eqnarray}}

\newcommand{\R}{\mathbb{R}}

\newcommand{\norm}[1]{\left\lVert{#1}\right\rVert}

\usepackage{tikz}	
\usepackage[hang,small,bf]{caption}  
\usetikzlibrary{backgrounds,fit,decorations.pathreplacing}  
\makeatletter

\@ifundefined{textcolor}{}
{%
 \definecolor{BLACK}{gray}{0}
 \definecolor{WHITE}{gray}{1}
 \definecolor{RED}{rgb}{1,0,0}
 \definecolor{GREEN}{rgb}{0,.4,0}
 \definecolor{BLUE}{rgb}{0,0,1}
 \definecolor{CYAN}{cmyk}{1,0,0,0}
 \definecolor{MAGENTA}{cmyk}{0,1,0,0}
 \definecolor{YELLOW}{cmyk}{0,0,1,0}
 }

\makeatother

\usepackage{bm}
\usepackage{mathtools}

\usepackage{graphicx}  
\usepackage{dcolumn}   
\usepackage{bm}        
\usepackage{color}
\usepackage[]{algorithm}
\usepackage{braket}
\usepackage{bbm}
\usepackage{mathtools}
\usepackage{amstext,amsmath,amssymb,ulem,amsthm}
\usepackage{color}

\usepackage{times}
\usepackage{algorithm}
\floatname{algorithm}{Algorithm}
\usepackage{ftnxtra}

\newcommand{\Ord}[1]{\mathcal{O}\left(#1\right)}

\renewcommand{\braket}[2]{\left< #1 \left|#2 \right.\right>}

\newtheorem{thm}{Theorem}
\newtheorem{lem}{Lemma}

\usepackage[T1]{fontenc}
\usepackage[utf8]{inputenc}
\usepackage{graphicx}
\usepackage[colorlinks]{hyperref}

\def\id{I}

\def\1{\mat{\id}}
\def\mat#1{\mathbf{#1}}

\renewcommand{\vec}[1]{\bm{\mathrm{#1}}}

\renewcommand{\sout}[1]{}

\begin{document} 

\title{A quantum linear system algorithm for dense matrices}
                                                       
\author{Leonard Wossnig}
\affiliation{Theoretische Physik, ETH Z\"urich}
\affiliation{Department of Materials, University of Oxford}

\author{Zhikuan Zhao}
\email{zhikuan_zhao@mymail.sutd.edu.sg}
\affiliation{Singapore University of Technology and Design}
\affiliation{Centre for Quantum Technologies, National University of Singapore}

\author{Anupam Prakash}
\affiliation{Centre for Quantum Technologies, National University of Singapore}

\date{\today}
\begin{abstract}
Solving linear systems of equations is a frequently encountered problem in machine learning and optimisation. Given a matrix $A$ and a vector
$\mathbf b$ the task is to find the vector $\mathbf x$ such that $A \mathbf x = \mathbf b$. We describe a quantum algorithm 
that achieves a sparsity-independent runtime scaling of $\Ord{\kappa^2 \norm{A}_F \cdot\text{polylog}(n)/\epsilon}$, where $n\times n$ is the dimensionality of $A$ with Frobenius norm $\norm{A}_F$, $\kappa$ denotes the condition number of $A$, and $\epsilon$ is the desired precision parameter. When applied to a dense matrix with spectral norm bounded by a constant, the runtime of the proposed algorithm is bounded by $\Ord{\kappa^2\sqrt{n} \cdot\text{polylog}(n)/\epsilon}$, which is a quadratic 
improvement over known quantum linear system algorithms.
Our algorithm is built upon a singular value estimation subroutine, which makes use of a memory architecture that allows for efficient preparation of quantum states that correspond to the rows and row Frobenius norms of $A$.
\end{abstract}

\maketitle
 
\paragraph{Introduction.}
A common bottleneck in statistical learning and machine learning algorithms is the inversion of high-dimensional matrices in order to solve linear systems of equations. Examples include covariance matrix inversions in Gaussian processes and support vector machines, as well as data matrix inversions in large scale regression problems~\cite{Rasmussen2004,Bishop2006}. 

Recent advances in the field of quantum information processing have provided promising prospects for the efficient solution of high-dimensional linear systems. The breakthrough work of Harrow, Hassidim and Lloyd \cite{Harrow2009a} introduced the quantum linear system algorithm (QLSA) that computes the quantum state  $\ket{\mathbf{x}}=\ket{A^{-1}\mathbf{b}}$ corresponding to the solution of a linear system $A \mathbf{x} = \mathbf{b}$, where $A\in \R^{n \times n}$ and $\mathbf x,\mathbf b \in \R^{n}$, in time $\Ord{\text{polylog}(n)}$ for a sparse and well-conditioned $A$. 
Unlike the output $A^{-1} \mathbf b  \in \R^{n}$ of a classical linear system solver, a copy of $\ket{A^{-1} \mathbf b }$ does not provide access to the coordinates of $A^{-1} \mathbf b $. Nevertheless, it allows us to perform useful computations such as sampling from the solution vector. The QLSA algorithm has inspired several works 
\cite{Rebentrost2013,Schuld2016,Wiebe2012,Wiebe2015b,Wiebe2015,Wiebe2014,Wiebe2014a,Wiebe2016,Zeng2016,Zhao2015}
in the emerging research area of quantum machine learning.

In the classical setting, the best known algorithm for the sampling task performed by the QLSA algorithm requires solving the 
linear system. The running time for a classical linear system solver scales as $\Ord{n^{\omega}}$, where the matrix multiplication exponent $\omega \le 2.373$~\cite{Coppersmith1990, L14}. 
However, as the sub-cubic scaling is difficult to achieve in practice, linear system solvers typically use the Cholesky decomposition and require time $\Ord{n^{3}}$ for dense matrices.

The QLSA algorithm \cite{Harrow2009a} has running time $\tilde{O}(\kappa^{2} s(A)^{2}/\epsilon)$ where $\kappa$ is the condition number, 
$s(A)$ is the sparsity or the maximum number of non-zero entries in a row of $A$ and $\epsilon$ is the precision to which the solution is approximated. 
There have been several improvements to the QLSA algorithm since the original proposal that have improved the running time to linear in $\kappa$ 
and $s(A)$ and to poly-logarithmic in the precision parameter $\epsilon$ \cite{Childs2015, Ambainis2010}. The work \cite{Clader2013a} introduced pre-conditioning 
for the QLSA algorithm and extended its applicability.

Quantum machine learning is an emerging research area that attempts to harness the power of quantum information processing to obtain speedups for classical machine learning tasks. A number of quantum machine learning algorithms have been proposed \cite{Rebentrost2013,Schuld2016,Wiebe2012,Wiebe2015b,Wiebe2015,Wiebe2014,Wiebe2014a,Wiebe2016,Zeng2016,Zhao2015}. 
Most of these algorithms use a quantum linear system solver as a subroutine. However, as mentioned in \cite{Harrow2009a}, and later also pointed out in \cite{Childs2009, Aaronson2015b}, the QLSA potentially has a few caveats.
In particular, the QLSA achieves an exponential speedup over classical algorithms when the matrix $A$ is sparse
and well conditioned, due to the sparsity-dependent Hamiltonian simulation subroutine. The potential exponential advantage of QLSA is lost when it is applied to dense matrices, which consistute a large class of interesting applications. 
Examples include kernel methods,
~\cite{Wilson2015}, and artificial neural networks,
where particularly convolutional neural network architectures rely heavily on subroutines that manipulate large, non-sparse matrices~\cite{Yangqing2014,Chetlur2014}.
Alternative approaches to the quantum linear system problem that avoid sparsity dependence are therefore desirable for a more general application of quantum computing to classical learning problems.

In this letter we present a quantum algorithm for solving linear systems of equations using the quantum singular value estimation (QSVE) 
algorithm introduced in \cite{Kerenidis2016}. The proposed algorithm 
achieves a runtime $\Ord{\kappa^2 \norm{A}_F \cdot\text{polylog}(n)/\epsilon}$, where $\kappa$ denotes the condition number of $A$, $\norm{A}_F$ is the Frobenius norm and $\epsilon$ is the precision parameter. When the spectral norm $\norm{A}_*$ is bounded by a constant, the scaling becomes $\Ord{\kappa^2 \sqrt{n} \cdot \text{polylog}(n)/\epsilon}$, which amounts to a polynomial speed-up over the $\Ord{\kappa^2 n \cdot \text{polylog}(n)/\epsilon}$ scaling achieved by~\cite{Harrow2009a} when applied to dense matrices.

We start by introducing some preliminaries. For a symmetric matrix $A \in \R^{n\times n}$ with spectral decomposition $A=\sum_{i \in [n]}\lambda_i\mathbf{s}_i\mathbf{s}_i^{\dagger}$, the singular value decomposition is given by $A=\sum_{i}^{r}|\lambda_i| \mathbf{s}_i\mathbf{s}_i^{\dagger}$.  
We also need the well known quantum phase estimation algorithm:

\begin{thm}[Phase estimation \cite{Kitaev1995}]
  \label{pest}
  Let unitary $U \ket{v_j} = \exp{(i \theta_j)} \ket{v_j}$ with $\theta_j \in [ - \pi ,\pi ]$ for $j \in [n]$. There is a quantum algorithm that transforms $\sum_{j \in [n]} \alpha_j \ket{v_j} \to \sum_{j \in [n]} \alpha_j \ket{v_j} \ket{\overline{\theta}_j}$ such that $|\overline{\theta_{j}} - \theta_{j} |\leq \delta$ for all $j\in [n]$ with probability $1-1/\text{poly}(n)$ in time $\Ord{T_U \log{(n)} / \delta}$, where $T_U$ defines the time to implement $U$.
\end{thm}

\noindent Quantum singular value estimation can be viewed as an extension of phase estimation to non unitary matrices. It is the main algorithmic primitive required for our linear system solver. 
\begin{definition}[Quantum singular value estimation]
  \label{q_pe}
  Let $A \in \R^{m\times n}$ have singular value decomposition $A = \sum_{i} \sigma_{i} u_{i} v_{i}^{t}$. A quantum singular value estimation algorithm with precision $\delta$ transforms $\sum_{j \in [n]} \alpha_j \ket{v_j} \to \sum_{j \in [n]} \alpha_j \ket{v_j} \ket{\overline{\sigma}_j}$ such that $|\overline{\sigma_{j}} - \sigma_{j}|\leq \delta$ for all $j\in [n]$ with probability $1-1/\text{poly}(n)$. 
\end{definition}
\noindent A quantum singular value estimation (QSVE) algorithm with running time of $\tilde{O}(\norm{A}_{F} /\delta)$ was presented in \cite{Kerenidis2016}, where it was used for quantum recommendation systems. 
An SVE algorithm applied to a symmetric matrix estimates $|\lambda_{i}|$ but does not provide an estimate for $sign(\lambda_{i})$. 
However, in order to solve linear systems we also need to recover $sign(\lambda_{i})$. We provide a simple procedure for recovering the sign given an 
SVE algorithm. Our procedure provides a way to construct a quantum linear system solver from a QSVE algorithm in a black box manner. 

The main result of this letter is a quantum linear system solver based on the QSVE algorithm \cite{Kerenidis2016} 
that achieves a running time of $\Ord{\kappa^2 \norm{A}_F \cdot\text{polylog}(n)/\epsilon}$.
We briefly describe the QSVE algorithm in the next section. 
We then present the quantum linear system solver and provide a complete analysis for the linear system solver as well as a comparison with other approaches to the QLSA in the discussion.

\paragraph{The QSVE algorithm.} \label{s1} 

The QSVE algorithm requires the ability to efficiently prepare the 
quantum states corresponding to the rows and columns of matrix $A$. The matrix entries are stored in the following data structure, 
such that a quantum algorithm with access to this data structure has this ability. 
\begin{lem}[Data Structure~ \cite{Kerenidis2016}]
  \label{data_struc}
  Let $A \in \R^{m \times n}$ be a matrix with entries $A_{ij}$ which arrive in an arbitrary order. 
  There exists a data structure with the following properties:
  \begin{itemize}
  \item A quantum computer with access to the data structure can perform the following mappings in $\Ord{\text{polylog}(mn)}$ time.
  \begin{align}
U_\mathcal{M}: \ket{i}\ket{0}\rightarrow\ket{i,\vec{A_i}} &= \frac{1}{\|\vec{A_i}\|}\sum\limits_{j=1}^{n}A_{ij}\ket{i,j},\notag \\
U_\mathcal{N}: \ket{0}\ket{j}\rightarrow\ket{\vec{A}_F,j} &= \frac{1}{\|A\|_F}\sum\limits_{i=1}^{m}\|\vec{A_i}\|\ket{i,j},
  \end{align}
where $\vec{A_i}\in R^{n}$ correspond to the rows of the matrix $A$ and $\vec{A}_F\in \R^{m}$ is a vector whose entries are the $\ell_{2}$ norms 
of the rows, i.e.\ $(\vec{A}_F)_i=\norm{A_i}$. 
  \item The time required to store a new entry $A_{ij}$ is $\Ord{\text{log}^2(mn)}$ and data structure size is $O(w \log mn)$ where $w$ is the number of non zero entries in $A$. 
  \end{itemize}
\end{lem}
A possible realization of this data structure is based on an array of $m$ binary trees, each binary tree contains at most $n$ leaves which store the squared amplitudes of the corresponding matrix entry $|A_{ij}|^2$, as well as the sign of $A_{ij}$. An internal node of a tree stores the sum of the elements in the subtree rooted at it. The root of the $i^{th}$ tree then contains $\norm{\mathbf{A}_i} ^2, \, i \in [m]$. In order to access the vector of row Frobenius norms, we construct one more binary tree, the $i^{th}$ leaf of which stores $\norm{ \mathbf{A}_i} ^2$. A detailed description of such a binary tree memory structure, and the proof of Lemma~\ref{data_struc} can be found in \cite{Kerenidis2016}. 

The QSVE algorithm is a quantum walk based algorithm that leverages the connection between the singular values $\sigma_i$ of the target matrix $A$ and the principal angles $\theta_i$ between certain subspaces associated with $A$. The relation between quantum walks and eigenvalues has been well known in the literature and has been used in 
several previous results \cite{S04, C10}. However, the quantum walk defined by the QSVE algorithm is particularly interesting for linear systems as instead of the sparsity $s(A)$, it depends on the Frobenius norm $\norm{A}_{F}$.

The QSVE algorithm makes use of a factorization $\frac{A}{\|A\|_F}= \mathcal{M}^{\dagger}\mathcal{N}$, where $\mathcal{M}\in\mathbb{R}^{mn\times m}$ and $\mathcal{N}\in\mathbb{R}^{mn\times n}$ are isometries. The key idea is that the unitary operator $W$ defined by $W=(2\mathcal{M}\mathcal{M}^{\dagger}-I_{mn})(2\mathcal{N}\mathcal{N}^{\dagger}-I_{mn})$, where $I_{mn}$ is the identity matrix, can be implemented efficiently using the data structure in Lemma~\ref{data_struc}. Further, $W$ has two dimensional eigenspaces spanned by $\{\mathcal{M}\mathbf{u}_i,\mathcal{N}\mathbf{v}_i\}$ on which it acts as a roation by angle $\theta_i$, such that $\cos\frac{\theta_i}{2}=\frac{\sigma_i}{\|A\|_F}$. Note that the sub-space spanned by $\{\mathcal{M}\mathbf{u}_i,\mathcal{N}\mathbf{v}_i\}$ is therefore also spanned by $\{\mathbf{w}_i^+,\mathbf{w}_i^-\}$, the eigenvectors of $W$ with eigenvalues $\exp(i\theta_i)$ and $\exp(-i\theta_i)$ respectively. In particular we may write the following decomposition, $\ket{\mathcal{N}\mathbf{v}_i}=\omega_i^+\ket{\mathbf{w}_i^+}+\omega_i^-\ket{\mathbf{w}_i^-}$, with $|\omega_i^-|^2+|\omega_i^+|^2=1$.
Algorithm 1 describes the QSVE algorithm, the analysis is contained in the following lemma.

\begin{algorithm}[!htbp]
\caption{Quantum singular value estimation. \cite{Kerenidis2016}}
\label{algo1}
\begin{enumerate}
  \item Create the arbitrary input state $\ket{\alpha} = \sum_i \alpha_{\vec v_i} \ket{\vec v_i}$.
  \item Append a register $\ket{0^{\lceil \log{m} \rceil}}$ and create the state $\ket{\mathcal{N} \alpha} = \sum_i \alpha_{\vec v_i} \ket{\mathcal{N} \vec v_i}=\sum_i \alpha_{\vec v_i}(\omega_i^+\ket{\mathbf{w}_i^+}+\omega_i^-\ket{\mathbf{w}_i^-})$.
  \item Perform phase estimation \cite{Kitaev1995} with precision $2 \delta >0$ on input $\ket{\mathcal{N}\alpha}$ for $W=(2\mathcal{M}\mathcal{M}^{\dagger}-I_{mn})(2\mathcal{N}\mathcal{N}^{\dagger}-I_{mn})$ and obtain $\sum_i \alpha_{\vec v_i}(\omega_i^+\ket{\mathbf{w}_i^+,\overline{\theta}_i}+\omega_i^-\ket{\mathbf{w}_i^-,-\overline{\theta}_i})$, where $\overline{\theta_i}$ is
     the estimated phase $\theta_i$ in binary bit-strings.
  \item Compute $\overline{\sigma}_i = \cos{(\pm\overline{\theta_i}/2)}||A||_F$.
  \item Uncompute the output of the phase estimation and apply the inverse transformation of step (2) to obtain
    \begin{equation}
      \sum\limits_i \alpha_{\vec v_i} \ket{\vec v_i} \ket{\overline{\sigma_i}}
    \end{equation}
\end{enumerate}
\end{algorithm}

\begin{lem}[Preparation of the Isometries~\cite{Kerenidis2016}]
  \label{isometries}
  Let $A\in \R^{m\times n}$ be a matrix with singular value decomposition $A = \sum_i \sigma_i \vec u_i \vec v_i^{\dagger}$ stored in the data structure
  described in Lemma~\ref{data_struc}. Then there exist matrices $\mathcal{M} \in \R^{mn \times m}$, and $\mathcal{N} \in \R^{mn \times n}$, such that
  \begin{enumerate}
    \item $\mathcal{M}, \mathcal{N}$ are isometries, that is $\mathcal{M}^{\dagger}\mathcal{M} = I_m$ and $\mathcal{N}^{\dagger}\mathcal{N} = I_n$ such that $A$ can be factorized as $A/\norm{A}_F = \mathcal{M}^{\dagger}\mathcal{N}$. 
    
Multiplication by $\mathcal{M},\mathcal{N}$, i.e.\ the mappings $\ket{\alpha} \rightarrow \ket{\mathcal{M}\alpha}$ and $\ket{\beta} \rightarrow \ket{\mathcal{N}\beta}$ can be performed in time $\Ord{\text{polylog}(mn)}$, 
    \item The reflections $2\mathcal{M}\mathcal{M}^{\dagger}-I_{mn}$, $2\mathcal{N}\mathcal{N}^{\dagger}-I_{mn}$, and hence the unitary $W=(2\mathcal{M}\mathcal{M}^{\dagger}-I_{mn})(2\mathcal{N}\mathcal{N}^{\dagger}-I_{mn})$ can be implemented in time $\Ord{\text{polylog}(mn)}$.
    \item The unitary $W$ acts as rotation by $\theta_{i}$ on the two dimensional invariant subspace $\{\mathcal{M}\mathbf{u}_i,\mathcal{N}\mathbf{v}_i\}$ plane, such that $\sigma_i = \cos\frac{\theta_i}{2}\|A\|_F$, where $\sigma_{i}$ is the $i$-th singular value for $A$. 
  \end{enumerate}
\end{lem}

We outline the ideas involved in the analysis of the QSVE algorithm and refer to \cite{Kerenidis2016} for further details. The map $\mathcal{M}$ appends to an arbitrary input state vector $\ket{\alpha}$ a register that encodes the row vectors $\mathbf{A_i}$ of $A$, such that
\begin{align}
\mathcal{M}: \ket{\alpha}&=\sum\limits_{i=1}^{m}\alpha_i\ket{i}\rightarrow\sum\limits_{i=1}^{m}\alpha_i\ket{i,\vec{A_i}}=\ket{\mathcal{M}\alpha}. \notag 
\end{align}
The map $\mathcal{N}$ similarly appends to an arbitrary input state vector $\ket{\alpha}$ a register that encodes the vector $\vec{A_F}$ whose entries are the $\ell_{2}$ norms $\|\vec{A_i}\|_F$ of the rows of $A$, 
\begin{align}
\mathcal{N}: \ket{\alpha}=\sum\limits_{j=1}^{n}\alpha_j\ket{j}\rightarrow\sum\limits_{j=1}^{n}\alpha_j\ket{\vec{A}_F,j}=\ket{\mathcal{N}\alpha}. \notag 
\end{align}
The above maps can be efficiently implemented given the memory structure described by Lemma~\ref{data_struc}.

The factorization of $A$ follows from the amplitude encoding of $\vec{A_i}$ and $\vec{A}_F$. We have $\ket{i,\vec{A_i}}=\frac{1}{\|\vec{A_i}\|}\sum\limits_{j=1}^{n}A_{ij}\ket{i,j}$ and $\ket{\vec{A}_F,j}=\frac{1}{\|A\|_F}\sum\limits_{i=1}^{m}\|\vec{A_i}\|\ket{i,j}$, implying that $(\mathcal{M}^{\dagger}\mathcal{N})_{ij}=\braket{i,\vec{A_i}}{\vec{A}_F,j}=\frac{A_{ij}}{\|A\|_F}$. Similarly, it follows that $\mathcal{M}$ and $\mathcal{N}$ have orthonormal columns and thus $\mathcal{M}^{\dagger}\mathcal{M}=I_m$ and $\mathcal{N}^{\dagger}\mathcal{N}=I_n$.

To show the relation between the eigenvalues of $W$ and the singular values of $A$, we consider the following:
\begin{align}
W\ket{\mathcal{N}\vec{v}_i}=&(2\mathcal{M}\mathcal{M}^{\dagger}-I_{mn})(2\mathcal{N}\mathcal{N}^{\dagger}-I_{mn})\ket{\mathcal{N}\vec{v}_i}\nonumber\\
=&(2\mathcal{M}\mathcal{M}^{\dagger}-I_{mn})\ket{\mathcal{N}\vec{v}_i}\nonumber\\=&2\mathcal{M}\frac{A}{\|A\|_F}\ket{\vec{v}_i}-\ket{\mathcal{N}\vec{v}_i}\nonumber\\
=&\frac{2\sigma_i}{\|A\|_F}\ket{\mathcal{M}\vec{u}_i}-\ket{\mathcal{N}\vec{v}_i},
\end{align}
where we used the singular value decomposition $A=\sum_i\sigma_i\ket{\vec{u}_i}\bra{\vec{v}_i}$, and the fact that the right singular vectors $\{\vec{v}_i\}$ are mutually orthonormal. Note that $W$ rotates $\ket{\mathcal{N}\vec{v}_i}$ in the plane of $\{\mathcal{M}\mathbf{u}_i,\mathcal{N}\mathbf{v}_i\}$ by $\theta_i$, such that 
\begin{align}
	\cos\theta_i&=\bra{\mathcal{N}\vec{v}_i}W\ket{\mathcal{N}\vec{v}_i} \notag \\ 
	&=\frac{2\sigma_i}{\|A\|_F} \bra{\vec v_i} A^{\dagger} \ket{\vec u_i}-1 \notag \\
	&=\frac{2\sigma_i^2}{\|A\|_F^2}-1 ,
\end{align} where we have used the fact that $(2\mathcal{M}\mathcal{M}^\dagger-I_{mn})$ represents a reflection in $\ket{\mathcal{M}\vec{u}_i}$
and that $A^{\dagger} = \mathcal{N}^{\dagger}\mathcal{M}= \sum_i \sigma_i \ket{\vec v_i} \bra{\vec u_i}$. 
Therefore the angle between $\ket{\mathcal{N}\vec{v}_i}$ and $\ket{\mathcal{M}\vec{u}_i}$ is given by $\frac{\theta_i}{2}$, i.e.\ half of the total rotation angle. Comparing the above expression with the half-angle formula for cosine functions, we obtain the relation $\cos\left(\frac{\theta_i}{2}\right)=\frac{\sigma_i}{\|A\|_F}$.

The two dimensional sub-space spanned by $\{\mathcal{M}\mathbf{u}_i,\mathcal{N}\mathbf{v}_i\}$ is therefore invariant under the action of $W$ which acts on it as a rotation by angle $\theta_i$.
The corresponding eigenvectors of $W$ have hence eigenvalues $\exp( \pm i\theta_i)$, and in particular we can perform phase estimation to get an estimate $\pm \overline{\theta_{i}}$ and then compute $\overline{\sigma_{i}} = \cos(\overline{\theta_{i}}/2) \norm{A}_{F}$ to estimates the singular values. We therefore have established the correctness of Algorithm~\ref{algo1} and have the following theorem, 

\begin{thm}[Quantum Singular Value Estimation~\cite{Kerenidis2016}]
  \label{qsve_thm}
  Let $A \in \R^{m \times n}$ be a matrix with singular value decomposition $A = \sum_i \sigma_i \vec u_i \vec v_i^{\dagger}$ stored in the data structure
 in Lemma~\ref{data_struc}. Further let $\delta > 0$ be the precision number. There is an algorithm that runs in
  $\Ord{\text{polylog}(mn)/\delta}$ and performs the mapping $\sum_i \alpha_{\vec v_i} \ket{\vec v_i} \ket{0} \rightarrow \sum_i \alpha_{\vec v_i} \ket{\vec v_i}\ket{\overline{\sigma_i}}$ where $\overline{\sigma_i} \in \sigma_i \pm \delta \norm{A}_F$ for all $i$ with probability at least $1-1/\text{poly}(n)$.
\end{thm}

The runtime of QSVE is dominated by the phase estimation procedure which returns an $\delta$-close estimate of $\theta_i$, s.t.\ $|\overline{\theta}_i -\theta_i| \leq 2 \delta$, which translates into the estimated singular value via $\overline{\sigma}_i = \cos{(\overline{\theta}_i/2)} \norm{A}_F$. 
The error in $\sigma_i$ can then be bounded from above by $|\overline{\sigma}_i - \sigma_i | \leq \delta \norm{A}_F$. 
The unitary $W$ can be implemented in time $\Ord{\text{polylog}(mn)}$ by Lemma~\ref{isometries}, by Theorem~\ref{pest} the running time for estimating of the singular values with additive error $\delta \norm{A}_F$ 
in $\Ord{\text{polylog}(mn)/\delta}$.

\paragraph{Quantum linear system algorithm.} \label{s2} 

Without loss of generality we can assume that the matrix $A$ is Hermitian as it is well known that the general case can be reduced to the Hermitian 
case \cite{Harrow2009a}. A QSVE algorithm immediately yields a linear system solver for positive definite matrices as the estimated singular values and eigenvalues are related via $\overline{\sigma_i}=|\overline{\lambda_i}|$. 
In order to solve general linear systems we need to recover the sign of each $\overline{\lambda_i}$. 
We provide a simple algorithm that recovers the signs using the QSVE procedure as a black box incurring only a constant overhead over the QSVE. 

\begin{algorithm}[!hbb]
  \caption{Quantum linear system solver.}
  \label{algo2}
  \begin{enumerate}
    \item Create the state $\ket{\mathbf{b}} = \sum_i \beta_{i}  \ket{\vec v_i}$ with $\vec v_i$ being the singular vectors of $A$.
    \item Perform two QSVEs as in Algorithm~\ref{algo1} for matrices $A, A+ \mu I$ with precision $\epsilon=1/\kappa $ 
    where $\mu=4/\kappa$ to obtain 
    $$\sum_i \beta_{i}  \ket{\vec v_i}_A \ket{|\overline{\lambda}_i|}_B \ket{|\overline{\lambda}_i + \mu|}_C.$$
    \item Add an auxiliary register and set it to $1$ if the value in register $B$ is greater than that in register $C$ and apply a conditional
    phase gate:
    $$\sum_i (-1)^{f_i} \beta_{i}  \ket{\vec v_i}_A \ket{|\overline{\lambda}_i|}_B \ket{|\overline{\lambda}_i + \mu|}_C \ket{f_i}_D.$$
    \item Add an ancilla register and apply a rotation conditioned on register $B$ with $\gamma=O(1/\kappa)$.  
    Then uncompute the registers $B,C, D$ to obtain
    \begin{align*}
    \sum_i(-1)^{f_i} \beta_{i} \ket{\vec v_i}\left( \frac{  \gamma}{\overline{\lambda_i}} \ket{0} + \sqrt {1 - \frac{  \gamma}{\overline{\lambda_i}}^{2} }  \ket{1} \right ) 
    \end{align*}
    Post-select on the ancilla register being in state $\ket{0}$.
  \end{enumerate}
\end{algorithm}
\noindent  The main result of this letter is the following theorem.
\begin{thm}
\label{thm_qlsa}
 	Let $A \in \R^{n \times n}$ be a Hermitian matrix with spectral decomposition 
 	$A = \sum_i \lambda_i \mathbf{\vec u_i \vec u_i}^{\dagger}$ stored in the data structure 
 	in Lemma~\ref{data_struc}. Further let $\kappa$ be the condition number $A$, and $\norm{A}_F$ the Frobenius norm and $\epsilon > 0$ be a precision parameter. Then Algorithm~\ref{algo2} has runtime $\Ord{\kappa^2 \cdot \text{polylog}(n) \cdot \norm{A}_F/\epsilon}$ that outputs state 
 	 $\ket{\overline{A^{-1} \mathbf b}}$ 
 	such that $\norm{ \ket{ \overline{ A^{-1} \mathbf b } } - \ket{ A^{-1} \mathbf b} } \leq \epsilon $. 
 	\end{thm}
	
\begin{proof}

	We first argue that Algorithm~\ref{algo2} correctly recovers the sign of the $\lambda_{i}$. The algorithm compares the estimates obtained by performing QSVE for $A$ and for $A^\prime = A+\mu I_{n}$, where $\mu$ is a positive scalar to be chosen later. 
	The matrix $A^\prime$ has the same eigenvectors as $A$, but has eigenvalues $\lambda_i + \mu$. Note that for $\lambda_{i} \geq 0$ we have $|\lambda_i + \mu| = |\lambda_i| + 
	|\mu| \geq |\lambda_{i}| $, however if $\lambda_{i} \leq -\mu/2 $ then $|\lambda_{i} + \mu| \leq |\lambda_{i}|$. 
	
	 Thus, if the estimates were perfect, then choosing $\mu = 2/\kappa$ would recover the sign correctly when the eigenvalues of $A$ lie in the interval $[-1, -1/\kappa] \cup [ 1/\kappa, 1]$. With the choice $\mu=4/\kappa$ and $\epsilon= 1/\kappa$ we find that the signs are still correct for all $\lambda_{i}$.

	We outline the derivation of the runtime, the analysis of the error bounds appears in the appendix. 
	The running time for QSVE with precision $\epsilon = 1/\kappa$ is $\tilde{O}( \kappa \norm{A}_{F})$. 
  Considering the success probability of the post-selection step, we require on average $\Ord{\kappa^2}$
	repetitions of the coherent computation. This can be reduced to $\Ord{\kappa}$ using amplitude amplification~\cite{Brassard2002}.
	Therefore an upper-bound of the runtime of our algorithm is given by $\Ord{\kappa^2 \cdot \text{polylog}(n) \norm{A}_F /\epsilon}$.
	\end{proof} 
	 
\paragraph{Discussion.}

The error dependence on the Frobenius norm suggests that our algorithm is most accurate when the $\norm{A}_F$ is bounded by some constant, in which case the algorithm returns the output state with a constant $\epsilon$-error in polylogarithmic time even if the matrix is non-sparse. More generally, as in the QLSA algorithm we can assume that the spectral norm $\norm{A}_*$ is bounded by a constant, although the Frobenius norm may scale with the dimensionality of the matrix. In such cases we have $\norm{A}_F=\Ord{\sqrt{n}}$. Hence in such scenarios the proposed algorithm runs in $\Ord{\kappa^2\sqrt{n} \cdot\text{polylog}(n)/\epsilon}$ and returns the output with a constant $\epsilon$-error. 

It was shown in~\cite{berry2009black} that given black-box access to the matrix elements, Hamiltonian simulation with error $\delta_h$ can be performed in time $\Ord{n^{2/3}\cdot\text{polylog}(n)/\delta_h^{1/3}}$.
This leads to a linear system algorithm based on~\cite{Harrow2009a} which scales as $\Ord{\kappa^2 n^{2/3} \cdot\text{polylog}(n)/\epsilon}$, where we have assumed the dominant error comes from phase estimation, and hence the error introduced by the Hamiltonian simulation is neglected. It was also shown numerically that the method of~\cite{berry2009black} attains a typical scaling of $\Ord{\sqrt{n}\cdot\text{polylog}(n)/ \delta_h^{1/2}}$ when applied to randomly selected matrices, leading to a $\Ord{\kappa^2\sqrt{n} \cdot\text{polylog}(n)/\epsilon}$ linear system algorithm. The work ~\cite{berry2009black} assumes that we have black-box access to the matrix entries, that is quantum queries of the form $\ket{i,j, 0} \to \ket{i,j,A_{ij}}$ are allowed. We note that in Lemma \ref*{data_struc} we instead assume black-box access to a data structure constructed in linear time from a stream of the matrix entries. Our access model thus differs from the one used in ~\cite{berry2009black} and therefore a direct comparison of the two results is not appropriate. The QSVE-based linear system solver achieves a $O(\sqrt{n})$-scaling with quantum access to 
the data structure in Lemma 1, and it is an interesting open question if one can achieve a similar scaling in the model with black box access to matrix entries. 

We also note that for practical implementations, the constant runtime overhead with respect to a given set of elementary fault-tolerant quantum gates is an important consideration. It has been shown by Scherer \textit{et al.}~\cite{Scherer2017} that current approaches to the QLSA potentially suffer from a large constant overhead, hindering prospects of near-term applications. Whether our proposed QSVE-based algorithm exhibits a more advantageous constant overhead due to the absence of Hamiltonian simulation, remains an open question.
\paragraph{Acknowledgements.}
The authors thank Simon Benjamin, Joseph Fitzsimons, Patrick Rebentrost and Nathan Wiebe for helpful comments on the manuscript, and Andrew Childs for his feedback on the earlier version. The authors also acknowledge support from Singapore’s Ministry of Education and National Research Foundation. This material is based on research funded in part by the Singapore National Research Foundation under NRF Award NRF-NRFF2013-01.

\bibliographystyle{apsrev}
\bibliography{qlsa_improved}

\section{Appendix}
\label{analysis}
In this appendix we establish error bounds on the final state.
The analysis below closely follows the analysis of the QLSA algorithm \cite{Harrow2009a}.

We use the filter functions $\mathrm{f}$ and $\mathrm{g}$ ~\cite{Hansen1998}, which allow us to invert only the the well-conditioned part of the matrix, i.e.\ the space which is spanned by the eigenspaces with eigenvalues, $\lambda_i\geq 1/\kappa$. We define the function 
$\mathrm{f}(\lambda) := 1/(\gamma\kappa \lambda)$ for $|\lambda| \geq 1/\kappa$, $\mathrm{f}(\lambda):=0$ for $\lambda \leq 1/2\kappa$, and $\mathrm{f}(\lambda)$ is a smooth interpolating function $\eta_{1}(\lambda)$ for $1/2\kappa \leq \lambda \leq 1/\kappa$. Similarly, we define $\mathrm{g}(\lambda) := 0$ for $|\lambda| \geq 1/ \kappa$,  $\mathrm{g}(\lambda) := 1/2$ for $\lambda \leq 1/2\kappa$, and $\mathrm{g}(\lambda)$ is an interpolating function $\eta_{2}(\lambda)$ for $1/2\kappa \leq \lambda \leq 1/\kappa$. The interpolating functions $\eta_{1}, \eta_{2}$ are chosen such that $f^{2}(\lambda) + g^{2}(\lambda) \leq 1$ for all $\lambda \in \R$. The algorithm in the main text corresponds to the choice $\mathrm{g}(\lambda)=0$.

Let $\gamma=\Ord{1/\kappa}$ be some constant which assures that the controlled rotation angle is less than $2\pi$ for any eigenvalues. 
We then define the map
\begin{multline}
        \ket{\mathrm{h}(\lambda)} := \sqrt{1-\mathrm{f}(\lambda)^2-\mathrm{g}(\lambda)^2} \ket{\mathrm{NO}} \\+ \mathrm{f}(\lambda) \ket{\mathrm{WC}} + \mathrm{g}(\lambda) \ket{\mathrm{IC}},
\end{multline}
with $\mathrm{f}^2(x) + \mathrm{g}^2(x) \leq 1$, where $\ket{\mathrm{NO}}$ indicates that no matrix inversion has taken place, $\ket{\mathrm{IC}}$ means that 
part of $\ket{\mathbf{b}}$ is in the ill-conditioned subspace of $A$, and $\ket{\mathrm{WC}}$ means that the matrix inversion has taken place and is
in the well conditioned subspace of $A$. 
This allows us to invert only the well conditioned part of the matrix while it flags the ill conditioned ones and 
interpolates between those two behaviours when $1/(2\kappa) < |\lambda| < 1/\kappa$. We therefore only invert eigenvalues which 
are larger than $1/(2\kappa)$, which motivates the choice of $\mu$ in Algorithm~\ref{algo2}.

Let $Q$ be the error-free operation corresponding to the QSVE subroutine followed by the controlled rotation without post-selection, i.e.\ 
\begin{align}
	\ket{\psi} := Q \ket{\mathbf b}\ket{0} \rightarrow \sum_i \beta_i \ket{\mathbf  v_i} \ket{\mathrm{h}(\lambda_i)}.
\end{align}
 $\overline{Q}$ in contrast describes the same procedure but the phase estimation step is erroneous, i.e.\
 \begin{align}
 	\ket{\overline{\psi}} := \overline{Q} \ket{\mathbf  b}\ket{0} \rightarrow \sum_i \beta_i \ket{\mathbf v_i} \ket{\mathrm{h}(\overline{\lambda}_i)}.
 \end{align}

We want to bound the error in $\norm{ \overline{Q} - Q}$. By choosing a general state $\ket{\mathbf b}$, this is equivalent to the bound
in $\norm{ Q \ket{\mathbf b} - \overline{Q} \ket{\mathbf b}} := \norm{ \ket{\overline{\psi}} - \ket{\psi} }$. We will make use of the following lemma.
\begin{lem}[\cite{Harrow2009a}]
\label{lipschitz}
	The map $\lambda \rightarrow \ket{h(\lambda)}$ is $\Ord{\kappa}$-Lipschitz, i.e.\ $\forall \lambda_i \neq \lambda_j$:
	\begin{align}
		\norm{ \ket{\mathrm{h}(\lambda_i)} - \ket{\mathrm{h}(\lambda_j)} } 
		\leq c \kappa | \lambda_i - \lambda_j|,
	\end{align}
	for some $c  \leq \pi/2 = \Ord{1}$. 
\end{lem}

 As $\norm{ \ket{\overline{\psi}} - \ket{\psi} } = \sqrt{2 \left(1-Re \braket{\overline{\psi}}{\psi} \right)},$ it suffices 
to lower-bound $Re\braket{\overline{\psi}}{\psi}$:
 \begin{align}
 	Re \braket{\overline{\psi}}{\psi} = \sum\limits_{i=1}^N |\beta_i |^2  Re\braket{\mathrm{h}(\overline{\lambda}_i)}{\mathrm{h}(\lambda_i)}\notag \\
 	\geq \sum\limits_{i=1}^N |\beta_i |^2 \left( 1 - \frac{c^2 \kappa^2 \delta^2 \norm{A}_F^2}{2} \right),
 \end{align}
where we used the error bounds of the QSVE subroutine for the eigenvalue distance, i.e. $| \lambda_i - \overline{\lambda}_i | \leq \delta \norm{A}_F$, which is a consequence of the phase estimation accuracy, and the $\Ord{\kappa}$-Lipschitz property
 in Lemma~(\ref{lipschitz}).
Since $ 0 \leq Re \braket{\overline{\psi}}{\psi} \leq 1$, it follows that 
\begin{align}
	1 - Re \braket{\overline{\psi}}{\psi} \leq 
	\sum\limits_{i=1}^N |\beta_i |^2 \left(\frac{c^2 \kappa^2 \delta^2 \norm{A}_F^2}{2} \right) 
\end{align}
Using $\sum_i |\beta_i|^2 = 1$, the distance can be bounded as 
\begin{align}
	\norm{ \ket{\overline{\psi}} - \ket{\psi} } \leq \Ord{\kappa \delta \norm{A}_F}.
\end{align}
If we require this error to be of $\Ord{\epsilon}$,
we need to take the phase estimation accuracy to be $\delta = \Ord{\frac{\epsilon}{\kappa \norm{A}_F}}$. 
This results in a runtime $\Ord{\kappa \norm{A}_F \cdot \text{polylog}(n)/\epsilon}$.
In order to successfully perform the post-selection step, we need to repeat the algorithm on average $\kappa^2$ times. This additional multiplicative factor of $\kappa^2$ can be reduced to $\kappa$ using amplitude amplification~\cite{Brassard2002}.
Putting everything together, we have a final runtime of $\Ord{\kappa^2 \norm{A}_F \cdot \text{polylog}(n)/\epsilon}$.
\end{document}